%% file: conference_101719.tex
\documentclass[conference]{IEEEtran}
\IEEEoverridecommandlockouts
\usepackage{cite}
\usepackage{amsmath,amssymb,amsfonts,amsthm}
\usepackage{array}
\usepackage[caption=false,font=footnotesize]{subfig}
\usepackage{algorithm}
\usepackage{algpseudocode}
\usepackage{graphicx}
\usepackage{textcomp}
\usepackage{xcolor}
\usepackage{optidef}
\def\BibTeX{{\rm B\kern-.05em{\sc i\kern-.025em b}\kern-.08em
    T\kern-.1667em\lower.7ex\hbox{E}\kern-.125emX}}

\input{defs}

\begin{document}

\title{Consensus-Based Set-Theoretic Control in Power Systems\\
\thanks{This work is partially supported by the National Science Foundation Graduate Research Fellowship Program under Grant No. DGE-1762114 and NSF grant ECCS-1930605. Any opinions, findings, conclusions, or recommendations expressed in this material are those of the authors and do not necessarily reflect the views of the National Science Foundation.}
}

\author{\IEEEauthorblockN{Daniel Tabas}
\IEEEauthorblockA{\textit{Dept. of Electrical and Computer Engineering} \\
\textit{University of Washington}\\
Seattle, WA, United States\\
dtabas@uw.edu}
\and
\IEEEauthorblockN{Baosen Zhang}
\IEEEauthorblockA{\textit{Dept. of Electrical and Computer Engineering} \\
\textit{University of Washington}\\
Seattle, WA, United States \\
zhangbao@uw.edu}
}

\maketitle

\begin{abstract}

Set-theoretic control is a useful technique for dealing with the uncertainty introduced into power systems by renewable energy resources. Although set operations are computationally expensive in large systems, distributed approaches serve as a remedy. In this paper, we propose a novel consensus-based approach for set-theoretic frequency control in power systems. A robust controlled-invariant set (RCI) for the system is generated by composing RCIs for each bus in the network. The process of generating these sets uses a consensus-based approach in order to facilitate discovery of mutually compatible subsystem RCIs. Each bus seeks to maximize the size of its own RCI while treating the effects of coupling as an unknown-but-bounded disturbance. The consensus routine, which demonstrates linear convergence, is embedded into a backwards reachability analysis of initial safe sets. Results for a 9-bus test case show that simple model predictive controllers associated with the resulting RCIs maintain safe operation when the system is subjected to worst case (adversarial) fluctuations in net demand, where conventional controllers are shown to fail.
\end{abstract}


\input{Sections/section1}
\input{Sections/section2}
\input{Sections/section3}
\input{Sections/section4}
\input{Sections/section5}

\input{Sections/section6}
\input{Sections/section7}
\bibliographystyle{ieeetr}
\bibliography{references,ref2}
\input{Sections/section8}

\end{document}

%% file: defs.tex
\newcommand{\R}{\mathbb{R}}
\newcommand{\any}{\text{ $\forall$ }}
\newcommand\m[1]{\begin{bmatrix}#1\end{bmatrix}}

\algnewcommand\algorithmicforeach{\textbf{for each}}
\algdef{S}[FOR]{ForEach}[1]{\algorithmicforeach\ #1\ \algorithmicdo}

\newcommand{\X}{\mathcal{X}}
\newcommand{\D}{\mathcal{D}}
\newcommand{\U}{\mathcal{U}}
\renewcommand{\P}{\mathcal{P}}
\newcommand{\M}{\mathcal{M}}
\renewcommand{\r}{\mathcal{R}}
\newcommand{\Xinf}{\mathcal{X}^{(\infty)}}
\newcommand{\Y}{\mathcal{Y}}
\newcommand{\N}{\mathcal{N}}
\newcommand{\A}{\mathcal{A}}
\newcommand{\B}{\mathcal{B}}
\renewcommand{\S}{\mathcal{S}}

\newcommand{\T}{\mathcal{T}}
\newcommand{\Q}{\mathcal{Q}}

\theoremstyle{definition}
\newtheorem{definition}{Definition}

\newtheorem{theorem}{Theorem}
\newtheorem*{theorem*}{Theorem}

\newtheorem{lemma}{Lemma}

\DeclareMathOperator*{\argmin}{arg\,min}

%% file: Sections/section1.tex
\section{Introduction}
The power system is a quintessential example of safety-critical infrastructure. Safe operation of power systems involves remaining in a safe region (or a safe set) defined, for example, by a range of acceptable voltages or frequencies at each bus. Conventionally, the credible contingencies that can lead to power system failure are enumerated and simulated to ensure that the current operating point is safe \cite{Bumby2009}. These contingencies are discrete events such as line or generator faults, while continuous changes in load or generation are assumed to be small and do not impact the stability of the grid. However, with the growth in variable renewable energy resources, fluctuations in net demand are increasing while system mechanical inertia is decreasing. Therefore, in addition to faults, fluctuations in renewable energy supply can cause stability issues for the system \cite{Kroposki2017}.


Uncertainties in net load can be viewed as disturbances in power injections. Unfortunately, it is difficult to provide probabilistic descriptions of these disturbances, since they are often correlated spatially and do not follow simple distributions~\cite{Chen18}. In contrast, the bounds on the disturbances are much easier to determine from physical system characteristics and data~\cite{Zhao18}. Therefore, a different point of view is to assume that disturbances take an ``unknown-but-bounded'' form. 

Even though renewable resources lead to larger disturbances, they also provide greater control capabilities. Power electronic inverters can realize a wide range of control actions subject to power and energy constraints. Therefore, the question becomes how to best design bounded controllers simultaneously for all possible bounded disturbances. Set-based control arises naturally as a solution to this problem. The objective is to keep the system states in a safe set subject to disturbances that can take on unknown but bounded values \cite{Blanchini2015}. 




The main challenge in using set-based methods is the computational cost of finding a safe set and its associated controller. One common approach is to propagate sets forwards or backwards through the system dynamics in order to draw conclusions about reachability, safety, and set invariance \cite{Blanchini2015,Rakovic2008,Trodden2016}. This method scales poorly with the system's dimension, making it nontrivial to apply to power systems composed of many synchronous generators and distributed energy resources. The standard technique to resolve this ``curse of dimensionality'' is to restrict the safety sets to be ellipsoids \cite{Kurzhanskiy2007} or zonotopes \cite{Althoff2011}. While such methods are computationally efficient, the resulting sets tend to be overly conservative. 



This paper addresses the computational challenge of using set-theoretic control for power system frequency stability. Specifically, we consider linearized swing equations and present a model predictive control (MPC) method that guarantees safe operation. By characterizing the control invariant set of the system, we show that simple, completely decentralized MPCs can be solved to ensure that the frequency and angle at a bus never leave a prescribed operating region. We do this through two contributions. The first is that we adopt a \emph{distributed approach}, where the computation of safe sets is viewed as a consensus problem. Therefore only two-dimensional sets (frequency and angle) need to be considered at each iteration, and they can be handled using existing algorithms.  Second, we consider a controller that \emph{maximizes the stability region}. This maximizes robustness to initial conditions and disturbances and offers the best performance in the stability sense.


The compositional approach to computing safe sets is not new and has been adopted in, for example,~\cite{Riverso2013,Riverso2014,Nilsson2016,Rakovic2010a} and their references. However, these works often compute the minimum safe set rather than the maximum. A smaller safe set implies better performance in the form of, for example, smaller frequency deviations. But as the disturbances grow in magnitude, we argue it is also important to provide the maximum region within which the system can safely operate. 

The rest of the paper is organized as follows. Section \ref{sec:Preliminaries} presents notation and fundamental concepts. Section \ref{sec:Model} provides details about the power system model. Section \ref{sec:Sets} presents the algorithms used to compute controlled-invariant sets, while Section \ref{sec:Control} describes the design of the associated controller. Section \ref{sec:Simulations} presents simulation results for a 9-bus test case, where we show that without explicitly considering controlled invariance, conventional linear feedback and MPC controllers can lead to unsafe operations under bounded disturbances. 







%% file: Sections/section2.tex
\section{Preliminaries} \label{sec:Preliminaries}


We use the following notation throughout the paper. Assume $\A$ and $\B$ are convex, compact subsets of $\R^n$. The Minkowski sum of $\A$ and $\B$ is $\A \oplus \B = \{x + y \mid x \in \A , y \in \B\}.$ The Pontryagin difference of two sets $\A, \B \subseteq \R^n$ is $\A \ominus \B = \{x \in \R^n \mid x + y \in \A, \any y \in \B \}.$ Given $\{x_1,...,x_n\} \in \R^n$, projection onto a subset $\{x_i\}_{i\in \mathcal{I}}$ of the coordinates is given by $ Pr_\mathcal{I}(\A) = \{\{x_i\}_{i\in \mathcal{I}} \mid x \in \A\}.$ Set intersection is $ \A \cap \B = \{x \mid x \in \A, x \in \B\}. $ A linear transformation of a set is $F \A = \{Fx \mid x \in \A  \}.$ The evolution of a discrete time system with state $x$ and dynamics $f$ is written as $x^+ = f(x(t))$.

The following definitions introduce the concept of robust controlled-invariance which will be used throughout the paper. 

\theoremstyle{definition}
\begin{definition}[Invariant set \cite{Blanchini2015}] An invariant set for a system with dynamics $x^+ = f(x(t))$ is a set $\Q$ with the property 
    $x(t_0) \in \Q \implies x(t) \in \Q, \any t \geq t_0$.
\end{definition}

\theoremstyle{definition}
\begin{definition}[Robust controlled-invariant set \cite{Blanchini2015}] \label{def:RCI} Consider a system with dynamics $x^+ = f(x(t),u(t),d(t))$, where $x \in \R^n$ is the system state, $u \in \U \subseteq \R^m$ is a control input and $d \in \D \subseteq \R^p$ is an exogenous disturbance. The set $\Q$ is a \textit{robust controlled-invariant set} (RCI) for the system if there exists a feedback controller $u(x)$ such that if $x(t_0) \in \Q,$ then for all $t \geq t_0$ and all disturbance sequences $d(t) \in \D$, $x(t) \in \Q$.
\end{definition}


Now consider a set of safe operating conditions in a power system denoted by $\X \subseteq \R^n$. In light of Definition~\ref{def:RCI}, the purpose of computing an RCI is to ensure there exists a control law that guarantees safe operation for all time. Going further, the \emph{maximal RCI in} $\X$ permits safe operation under the widest possible range of operating conditions. The following definitions facilitate the computation of a maximal RCI. 

\theoremstyle{definition}
\begin{definition}[Preimage set \cite{Blanchini2015}] \label{def:preimage}
The \textit{preimage set} of a set $\X$ is $\r := \{x \in \R^n \mid \exists \ u \in \U : f(x(t),u(t),d(t) \in \X, \any d \in \D\}$.
\end{definition}

\theoremstyle{definition}
\begin{definition}[Admissible set \cite{Blanchini2015}] Given a safe set $\X$, the \textit{1-step admissible set} for $\X$ is given by $\X^{(1)} = \r \cap \X$, i.e. the set of states in $\X$ that can be kept inside $\X$ for one time step under any disturbance.
\end{definition}

We present an algorithm for computing maximal RCIs in Section \ref{sec:Sets_centralized} and extend it to the distributed case in Section \ref{sec:Sets_decentralized} using a consensus-based approach.

%% file: Sections/section3.tex
\section{Model} \label{sec:Model}

\subsection{System Dynamics}

Consider a power system with $N_B$ buses of which $N_L$ buses have loads and $N_G$ buses have synchronous generators. Each generator is assumed to be collocated with a distributed energy resource (e.g., solar or storage) that is capable of modulating a portion of its power output or consumption in order to provide frequency regulation. This assumption also applies to inverter-based systems that provide synthetic or virtual inertia. 

Let $\mathcal{G}$ be the set of buses with generators, and let $\mathcal{L}$ be the set of buses with loads. For $i \in \mathcal{G},$ generator $i$ is modeled as an emf $e_i \angle (\delta_i^0 + \delta_i)$ behind synchronous winding reactance $X_{\text{d}i}$, where $\delta_i^0$ is the steady-state rotor angle and $\delta_i$ is the deviation from $\delta_i^0$. The voltage at the generator terminals is $V_i \angle (\theta_i^0 + \theta_i)$, with $\theta_i^0$ and $\theta_i$ defined similarly. The variables $u_i$ for $i \in \mathcal{G}$ and $d_i$ for $i \in \mathcal{L}$ are the control input (power injection) and exogenous disturbance (net load fluctuation), respectively. From the linearized swing equation and DC power flow \cite{Kundur1994}, the dynamics at each bus are given by: \begin{subequations}\label{eqn:oct13_1}
    \begin{gather}
    0 = M_i \ddot{\delta}_i + D_i \dot{\delta}_i + K_i(\delta_i - \theta_i), i \in \mathcal{G} \label{eqn:oct13_2}\\
    \theta_i = \frac{u_{i\mid i \in \mathcal{G}} -d_{i\mid i \in \mathcal{L}} + \sum_{j \sim i} C_{ij}\theta_j + K_i \delta_{i\mid i \in \mathcal{G}}}{C_i + K_i} \label{eqn:nov8_1}
    \end{gather}
\end{subequations} where 
$M_i$ is the inertia,
$D_i$ is the damping coefficient, and
$K_i = \frac{e_iV_i}{X_{\text{d}i}}\cos(\delta_i^0 - \theta_i^0)$ for $i \in \mathcal{G}$ and $K_i = 0$ otherwise.
The constants $C_{ij}$ are given by $ V_iV_j[B_{ij}\cos(\theta_i^0-\theta_j^0) - G_{ij}\sin(\theta_i^0-\theta_j^0)]$, where $B_{ij}$ is the imaginary part of the $ij$-th entry of the bus admittance matrix $\mathbf{Y_{bus}}$ and $G_{ij}$ is the real part of the $ij$-th entry of $\mathbf{Y_{bus}}$, with $C_{ii} = 0.$ Lastly,
$C_i = \sum_{j \in \N_i} C_{ij},$ where $\N_i$ is the set of subsystems connected to bus $i$. This reflects the assumption that angle deviations are small but absolute angles may or may not be small.

The model in \eqref{eqn:oct13_1} uses a set of differential algebraic equations to describe the power system. In order to analyze the system in terms of set invariance, the algebraic equations are eliminated by finding a closed form solution to the linear DC power flow equations \eqref{eqn:nov8_1} and substituting for $\theta$ in the dynamic equations \eqref{eqn:oct13_2}. For the details of this procedure, see Appendix A1.

The system states are defined as follows. Recall $\delta_i$ and $\delta^0$ from \eqref{eqn:oct13_2}, and let $\omega_i = \dot{\delta}_i$ denote the deviation from nominal frequency $\omega^0$ at bus $i$. Let $x_i = \m{\delta_i & \omega_i}^T$, and let $x = \m{x_1^T & \cdots & x_{N_G}^T}^T \in \R^{2N_G}.$ After substituting $\theta$ in \eqref{eqn:oct13_2}, the system dynamics in \eqref{eqn:oct13_1} can be rewritten using $x$ as $\dot{x} = Ax + Bu + Ed$ where $u = \m{u_{i_1}&\cdots&u_{i_{N_G}}}^T$ is a vector of control inputs (power injections) at the generator buses and $d =\m{d_{i_1}&\cdots&d_{i_{N_L}}}^T$ is a vector of exogenous disturbances (net load fluctuations) from the load buses. Details of the matrices $A$, $B$, and $E$ are given in Appendix A2.

It is important to distinguish between local dynamics and the effects of coupling because from the perspective of individual buses, all non-local phenomena can be conservatively treated as unknown-but-bounded disturbances, with the bounds coming from the safety sets of neighboring buses. For each of the $N_G$ subsystems, the dynamics are
\begin{align}
    \dot{x_i} = \underbrace{A_{1i}x_i + B_{1i}u_i}_\text{local dynamics} + \underbrace{A_{2i}y_i + B_{2i}u_{\N_i} + E_id}_\text{non-local effects}, \label{eqn:nov10_1}
\end{align} where $y_i$ is the states of the neighbors of bus $i$, and $u_{\N_i}$ is the control inputs from nodes neighboring $i$. The matrices $A_{1i}$ and $A_{2i}$ together constitute the nonzero blocks of the $i$th block-row of the matrix $A$, with $A_{1i}$ being the $i$th diagonal block of $A$. The matrix $B$ is similarly decomposed. The last three terms in \eqref{eqn:nov10_1} are referred to as the state coupling, input coupling, and exogenous disturbances, respectively.

The continuous dynamics from \eqref{eqn:nov10_1} are discretized via a first-order approximation with time step $h$, leading to the frequency dynamics at each generator bus being given by \begin{align} x_i^+=& \underbrace{\hat{A}_{1i}x_i(t) +\hat{B}_{1i}u_i(t)}_\text{local dynamics}
    + \underbrace{\hat{A}_{2i}y_i(t) + \hat{B}_{2i}u_{\N_i}(t) + \hat{E}_id(t)}_\text{non-local effects} \label{eqn:nov5_1}.
\end{align}

\subsection{System Constraints}
Let $\X_i\subseteq \R^2$ be the set of safe operating points $(\delta_i,\omega_i)$ for subsystem $i$, $\U_i \subseteq \R$ be the set of control actions available to $i$, and $\D \subseteq \R^{N_L}$ be the disturbance set for the system. These sets are given by
\begin{subequations} \label{eqn:nov5_2}
    \begin{gather}
    \X_i = \{(\delta_i,\omega_i) : |\delta_i| \leq \bar{\delta}_i, |\omega_i| \leq \bar{\omega}_i\}, \
    \U_i = \{u_i : |u_i| \leq \bar{u}_i\}\\
    \D = \{d : |d_i| \leq \bar{d}_i, i=1,...,N_L\}
    \end{gather}
\end{subequations} where $\bar{\delta}_i, \bar{\omega}_i, \bar{u}_i,$ and $\bar{d}_i$ are positive constants. Further, let $\X = \prod_{i=1}^{N_G}\X_i$ and $\U = \prod_{i=1}^{N_G}\U_i$.

Because of the presence of disturbances, not all points in $\X$ are invariant. That is, even though a point is in $\X$, there may exist a disturbance that can drive the states to be outside of $\X$. To ensure safe operation, we therefore need to find a \emph{robust controlled-invariant set (RCI)} inside of $\X$. The object of this paper, the \emph{maximal RCI}, ensures safe operation over the widest possible range of operating conditions.

We describe a system in terms of its dynamics and constraints using the notation $\Sigma = (\hat{A},\hat{B},\hat{E},\X,\U,\D)$ where $\hat{A},\hat{B},$ and $\hat{E}$ give the discrete-time dynamics and $\X,\U,$ and $\D$ are from \eqref{eqn:nov5_2}. Each subsystem of $\Sigma$, with discrete-time dynamics from \eqref{eqn:nov5_1}, is denoted $\Sigma_i = (\hat{A}_{1i},\hat{B}_{1i},\hat{A}_{2i},\hat{B}_{2i},\hat{E}_i,\X_i,\U_i,\D,\N_i)$.

%% file: Sections/section4.tex
\section{Computing robust controlled-invariant sets} \label{sec:Sets}

\subsection{Centralized computation} \label{sec:Sets_centralized}

Consider a system $\Sigma = (\hat{A},\hat{B},\hat{E},\X,\U,\D)$ and assume the matrix pair $(\hat{A},\hat{B})$ is controllable.
We seek to approximate the maximal RCI that fits inside the safe set $\X$. One method \cite{Blanchini2015} involves iterative backwards-reachability computations for the safe set. An approximation of the largest RCI, denoted $\S$, can be found in a finite number of iterations if the true RCI is nonempty \cite[Prop. 5.4]{Blanchini2015}. The procedure is summarized in Algorithm \ref{alg:backprop} and explained in detail in Appendix B.
Each step in the procedure is an operation on polytopes which can be performed using open-source software packages \cite{Heirung2019}.

\input{Algorithms/alg_backprop}

\subsection{Decentralized computation} \label{sec:Sets_decentralized}
The scalability of Algorithm \ref{alg:backprop} to higher dimensions is poor even with linear dynamics and polytopic constraint sets. Rather than directly computing the maximal RCI, we decompose the system into a network of subsystems and approximate the maximal RCI per subsystem, treating the effects of state and input coupling as additional bounded disturbances. The composition of each subsystem's RCI is a conservative (i.e., inner) approximation of the system-wide maximal RCI.

Consider a collection of interconnected subsystems $\{\Sigma_i\}_{i=1}^{N},$ each described by a tuple $ \Sigma_i = (\hat{A}_{1i},\hat{B}_{1i},\hat{A}_{2i},\hat{B}_{2i},\hat{E}_i,\X_i,\U_i,\D, \N_i)$ and assume each matrix pair $(\hat{A}_{1i},\hat{B}_{1i})$ is controllable. The algorithm for computing subsystem RCIs is a modification of Algorithm \ref{alg:backprop} applied to each subsystem $\Sigma_i$. In Algorithm \ref{alg:backprop}, lines 3-6 compute 1-step admissible sets backwards in time. Going forward in time, the coupling disturbances experienced by bus $i$ during the transition from $t = -k-1$ to $t = -k$ consist of the $k+1$-steps admissible sets of the buses connected to $i$. Therefore, the coupling disturbance sets $\{\X_j^{(k+1)}\}_{j \in \N_i}$ need to be known in order to compute the preimage of $\X_i^{(k)}$. This naturally leads to a fixed point algorithm for simultaneously computing $k+1$-steps admissible sets for all subsystems. The algorithm can be interpreted as a consensus routine where each bus iteratively compares information with its neighbors to reach an agreement on the state coupling disturbance sets at time $t = -k-1$ \cite{Mesbahi2010}.  In the controls literature, the consensus routine is often referred to as assume-guarantee reasoning \cite{Ghasemi2020,Chen2019b,Eqtami2019}. We summarize the consensus-based distributed computation of state coupling disturbance sets in Algorithm \ref{alg:consensus} and provide a convergence guarantee for sufficiently small discretization time step $h$ in Theorem 1. 

\input{Algorithms/alg_consensus}

\input{thm1}

\begin{proof}
The full proof is provided in Appendix \ref{app:thm} and sketched here. First, consider the response of a single subsystem $\Sigma_i$ to a change in the set of possible state coupling disturbances at time step $t = -k$, denoted $\Y_i^{(k)}$. We bound the resulting change in the size of $\X_i^{(k+1)}$ by a term that depends on the size of the discretization time step, the strength of the coupling between bus $i$ and its neighbors, and the size of the initial change. This result is extended to the case when all subsystems simultaneously experience a change to their individual sets of state coupling disturbances. If the time step $h$ is chosen appropriately small, then since $\Y_i^{(k,l+1)} = \prod_{j \in \N_i} \X_j^{(k+1,l)},$ the initial change in $\{\Y_i^{(k)}\}_{i=1}^{N_G}$ gives rise to a smaller subsequent change in $\{\Y_i^{(k)}\}_{i=1}^{N_G}$ through its effect on $\{X_j^{(k+1)}\}_{j=1}^{N_G}$. This process converges.  
\end{proof}

Distributed computation of maximal RCIs is a composition of Algorithms \ref{alg:backprop} and \ref{alg:consensus}. For each time step $k$, we run Algorithm \ref{alg:consensus} to find the true coupling disturbance sets. With this knowledge, an iteration of Algorithm \ref{alg:backprop} is performed. The procedure is summarized in Algorithm \ref{alg:dist_backprop}.

\input{Algorithms/alg_dist_backprop}


%% file: Algorithms/alg_backprop.tex
\begin{algorithm}
\caption{Centralized computation of largest RCI \cite{Blanchini2015}}
\begin{algorithmic}[1]

\Require $\Sigma, \varepsilon$
\Ensure $\S$ \Comment{Maximal approximate RCI}

\State $\X^{(k = 0)} = \X$

\For{$k = 0: k_\text{max}$}
\State $\P^{(k)} = \{x: x + Ed \in \X^{(k)}, \any d \in \D\} = \X^{(k)} \ominus E\D$
\State $\M^{(k)} = \{(x,u) \mid \hat{A}x + \hat{B}u \in \P^{(k)}; u \in \U\}$
\State $\r^{(k)} = \{x \mid \exists u \in \U: (x,u) \in \M^{(k)}\}$

\State $\X^{(k+1)}$ = $\r^{(k)} \cap \X = \r^{(k)} \cap \X^{(k)}$

\If{$(1+\varepsilon)\X^{(k+1)} \supseteq \X^{(k)}$}
\State Break \Comment{Success}
\ElsIf{$\X^{(k)} = \emptyset$}
\State Break \Comment{Failure}
\EndIf

\EndFor
\If{$k < k_\text{max}$}
\State $\S = \X^{(k+1)}$
\Else 
\State Inconclusive result
\EndIf

\end{algorithmic}
\label{alg:backprop}
\end{algorithm}

%% file: Algorithms/alg_consensus.tex
\begin{algorithm}
\caption{Consensus-based computation of state coupling disturbance sets}
\begin{algorithmic}[1]

\Require $\{\Sigma_i\}_{i=1}^{N}, k, \varepsilon$

\Ensure $\{\Y_i^{(k)}\}_{i=1}^N$
\State $\Y_i^{(k,l = 0)} = \prod_{j \in \N_i} \X_j^{(k)}$ for all $ i = 1,...,N$
\State $\U_{\N_i} = \prod_{j \in \N_i} \U_j$ for all $ i = 1,...,N$

\For{$l = 0:l_\text{max}$}
\For{$i = 1:N$}
\State $\P_i^{(k,l)} = \X_i^{(k)} \ominus \hat{A}_{2i} \Y_i^{(k,l)} \ominus \hat{E}_i \D \ominus \hat{B}_{2i}\U_{\N_i}$
\State $\M_i^{(k,l)} = \{(x_i,u_i) \mid u_i \in \U_i, \hat{A}_{1i}x_i + \hat{B}_{1i}u_i \in \P_i^{(k,l)}\}$
\State $\r_i^{(k,l)}$ = $\{x_i \mid \exists u_i: (x_i,u_i) \in \M_i^{(k,l)}\}$
\State $\X_i^{(k+1,l)} = \r_i^{(k,l)} \cap \X_i^{(k)}$
\State $\Y_i^{(k,l+1)} = \prod_{j \in \N_i} \X_j^{(k+1,l)}$

\EndFor
\If{$(1-\varepsilon)\Y_i^{(k,l)} \subseteq \Y_i^{(k,l+1)} \subseteq (1+\varepsilon) \Y_i^{(k,l)},\any i$}
\State Break \Comment{Converged}
\EndIf

\EndFor
\State $\Y_i^{(k)} = \Y_i^{(k,l+1)}$ for all $ i = 1,...,N$

\end{algorithmic}
\label{alg:consensus}
\end{algorithm}

\ifx
\begin{algorithm}
\caption{Consensus-based computation of coupling disturbance sets}
\begin{algorithmic}[h]

\Require $k,\varepsilon, \{A_{1i},A_{2i},B_{1i},B_{2i},E_i,\X_i^{(k)},\U_i\}_{i=1}^N, \D$
\Ensure $\Y^{(k)}$, flag\textunderscore $l$

\State $\Y^{(k,l = 0)} = \prod_{j = 1}^N \X_j^{(k)}$
\State $\U = \prod_{i=1}^{N} \U_i$
\State $l = 0$
\State flag\textunderscore $l$ = 0

\While{flag\textunderscore $l$ = 0}
\For{$i = 1:N$}
\State $\P_i^{(k,l)} = \X_i^{(k)} - A_{2i} \Y^{(k,l)}  - B_{2i}\U - E_i \D$
\State $\M_i^{(k,l)} = \{(x_i,u_i) \mid A_{1i}x_i + B_{1i}u_i \in \P_i^{(k,l)}; u_i \in \U_i\}$
\State $\r_i^{(k,l)}$ = $\{x_i \mid \exists u_i: (x_i,u_i) \in \M_i^{(k,l)}\}$
\State $\X_i^{(k+1,l)} = \r_i^{(k,l)} \cap \X_i^{(k)}$

\EndFor

\State $\Y^{(k,l+1)} = \prod_{j = 1}^N \X_j^{(k+1,l)}$

\If{$(1-\varepsilon)\Y^{(k,l)} \subseteq \Y^{(k,l+1)} \subseteq (1+\varepsilon) \Y^{(k,l)}$}
\State $\Y^{(k)} = \Y^{(k,l+1)}$
\State flag\textunderscore $l$ = 1 \Comment{Stop successfully}
\ElsIf{$\Y^{(k,l+1)} = \emptyset$}
\State $\Y^{(k)} = \emptyset$
\State flag\textunderscore $l$ = 2 \Comment{Stop unsuccessfully}
\Else
\State $l = l + 1$
\EndIf

\EndWhile

\end{algorithmic}
\label{alg:consensus}
\end{algorithm}

\fi

\ifx
\If{$(1-\varepsilon)\Y_i^{(k,l)} \subseteq \Y_i^{(k,l+1)} \subseteq (1+\varepsilon) \Y_i^{(k,l)}$}
\State $\Y_i^{(k)} = \Y_i^{(k,l+1)}$
\State flag\textunderscore $l$ = 1 \Comment{Stop successfully}
\ElsIf{$\Y_i^{(k,l+1)} = \emptyset$}
\State $\Y_i^{(k)} = \emptyset$
\State flag\textunderscore $l$ = 2 \Comment{Stop unsuccessfully}
\Else
\EndIf
\fi

%% file: thm1.tex
\begin{theorem} \label{thm:consensus_convergence}
Assume the dynamics of the system are as in \eqref{eqn:nov5_1} and that sets $\X_i$, $\U_i$, and $\D_i$ are as in \eqref{eqn:nov5_2}. Assume that the discretization time step $h$ is chosen according to
\begin{align}
    h \leq c \cdot \min_{i,j} \frac{1}{\|A_{2i}\|_2 \sqrt{2|\N_j|}}
\end{align} where $\|A_{2i}\|_2$ is the induced matrix 2-norm of the continuous-time coupling matrix $A_{2i}$, $|\N_j|$ is the degree of subsystem $j$, and $c$ is some positive constant. Then, the sets $\Y_i^{(k,l)}, i = 1,...,N_G$ in Algorithm \ref{alg:consensus} converge as $l \rightarrow \infty$.
\end{theorem}

%% file: Algorithms/alg_dist_backprop.tex
\begin{algorithm}
\caption{Distributed computation of subsystem RCIs}
\begin{algorithmic}[1]

\Require $\{\Sigma_i\}_{i=1}^{N}, \varepsilon$

\Ensure $\{\S_i\}_{i=1}^N$

\State $\X_i^{(k = 0)} = \X_i$ for all $i = 1,...,N$
\For{$k = 1:k_\text{max}$}
\State Compute $\Y_i^{(k)}$ for all $ i = 1,...,N$ \Comment{Alg. \ref{alg:consensus}}
\For{$i = 1:N$}
\State $\P_i^{(k)} = \X_i^{(k)} \ominus \hat{A}_{2i} \Y_i^{(k)} \ominus \hat{E}_i \D \ominus \hat{B}_{2i}\U_{\N_i}$
\State $\M_i^{(k)} = \{(x_i,u_i) \mid u_i \in \U_i, \hat{A}_{1i}x_i + \hat{B}_{1i}u_i \in \P_i^{(k)}\}$
\State $\r_i^{(k)}$ = $\{x_i \mid \exists u_i \in \U_i: (x_i,u_i) \in \M_i^{(k)}\}$
\State $\X_i^{(k+1)} = \r_i^{(k)} \cap \X_i^{(k)}$
\EndFor

\If{$(1+\varepsilon)\X_i^{(k+1)} \supseteq \X_i^{(k)},\any i$}
\State Break \Comment{Success}
\ElsIf{$\exists i: \X_i^{(k)} = \emptyset$}
\State Break \Comment{Failure}
\EndIf
\EndFor
\If{$k < k_\text{max}$}
\State $\S_i = \X_i^{(k+1)}$ for all $i = 1,...,N$
\Else 
\State Inconclusive result
\EndIf

\end{algorithmic}
\label{alg:dist_backprop}
\end{algorithm}

%% file: Sections/section5.tex
\section{Controller design} \label{sec:Control}

Once an RCI is generated, there exist many different control strategies that can render the set invariant \cite{Blanchini2015,Ames2019,Langson2004}. However, even simple myopic strategies demonstrate the efficacy of control based on RCIs. To demonstrate the concepts in this paper, a model predictive controller is assigned to each subsystem, with allowable control actions chosen from a \textit{regulation map} \cite{Blanchini2015} defined as  \begin{align}
    \Omega_i(x_i) = \{u_i \in \U_i \mid x_i^+ \in \S_i, \forall\ y_i \in \T_i, u_{\N_i} \in \U_{\N_i}, d \in \D\}
\end{align} where $\T_i = \prod_{j \in \N_i} \S_j$ is the state coupling disturbance set for bus $i$ and $\U_{\N_i} = \prod_{j \in \N_i} \U_j$ is the input coupling disturbance set for bus $i$. Since all relevant sets are polytopes, the set $\Omega_i(x_i)$ is a polytope which can be computed using set operations. By the definition of robust controlled-invariance and the fact that $\S_i \supseteq \Xinf_i$, the regulation map $\Omega_i(x_i)$ is guaranteed to be nonempty for all $x_i \in \Xinf_i$. This statement can be strengthened to hold for all $x_i \in \S_i$ if $\S_i$, the approximation of $\Xinf_i$, is invariant itself \cite{Rakovic2008}.

The local feedback control law is given by:
\begin{subequations}
    \begin{align}
    u_i(x_i) = &\argmin_{u_i \in \U_i} c(\hat{x}_i^+,u_i) \label{eqn:nov8_4}\\
    \text{ s.t. } &\hat{x}_i^+ = \hat{A}_{1i}x_i + \hat{B}_{1i}u_i, \label{eqn:nov8_5}\\ 
    &u_i \in \Omega_i(x_i). \label{eqn:nov8_6}
    \end{align}
\end{subequations} The cost $c$ in \eqref{eqn:nov8_4} is convex in $(\hat{x}_i^+,u_i)$, leading to a convex program. This ensures efficient computation and reflects realistic costs based on, for example, the 1-, 2-, or $\infty$-norm. The proposed controller will be referred to as the robust model predictive controller (RMPC). The invariance properties of the RCI $\S_i$ allow RMPC to outperform conventional controllers in the presence of adversarial attacks. We demonstrate this result on a 9-bus example in Section \ref{sec:Simulations}.

%% file: Sections/section6.tex
\section{Simulations} \label{sec:Simulations}

The RMPC is weighed against the following controllers.
\begin{enumerate}
    \item \textbf{One-step look ahead MPC with local feedback.} This controller only considers the safety set $\X_i$. The feedback control law is given by replacing \eqref{eqn:nov8_6} with $\hat{x}_i^+ \in \X_i$.
    \item \textbf{Infinite-horizon LQR with local feedback}. This controller is given by $u_i(x_i) = - \textbf{sat}(K_{\text{LQR}i}x_i),$ where $\textbf{sat}(\cdot)$ thresholds the output so that $u_i \in \U_i.$ The feedback gains $K_{\text{LQR}i}$ are computed for each bus individually via the solution to the algebraic Riccati equation.
\end{enumerate}




To demonstrate the superior robustness of RMPC, we compare its performance under worst-case disturbances to both the one-step look ahead MPC and the standard LQR. By worst case disturbance, we mean a system-wide exogenous disturbance whose goal is to drive subsystem $\Sigma_{\hat{i}}$ out of its safe set. We assume the disturbance knows the states $x_{\hat{i}}$ and control inputs $u_{\hat{i}}$ of subsystem $\Sigma_{\hat{i}}$, and chooses actions $d(t)$ from the set $\D$.


Results are generated for a 9-bus test case with three generators and three loads. Figure \ref{fig:sets} displays the safe and invariant regions of the state space per bus and the k-steps admissible sets that converge to the invariant sets. The safe set is defined as $|\delta| \leq 10^\circ, |\omega| \leq 0.6$ Hz. Algorithm \ref{alg:consensus} converges very fast empirically, generally taking about 3-4 iterations. 

\begin{figure}[ht]
    \centering
    \includegraphics[width = 9cm,height=3cm]{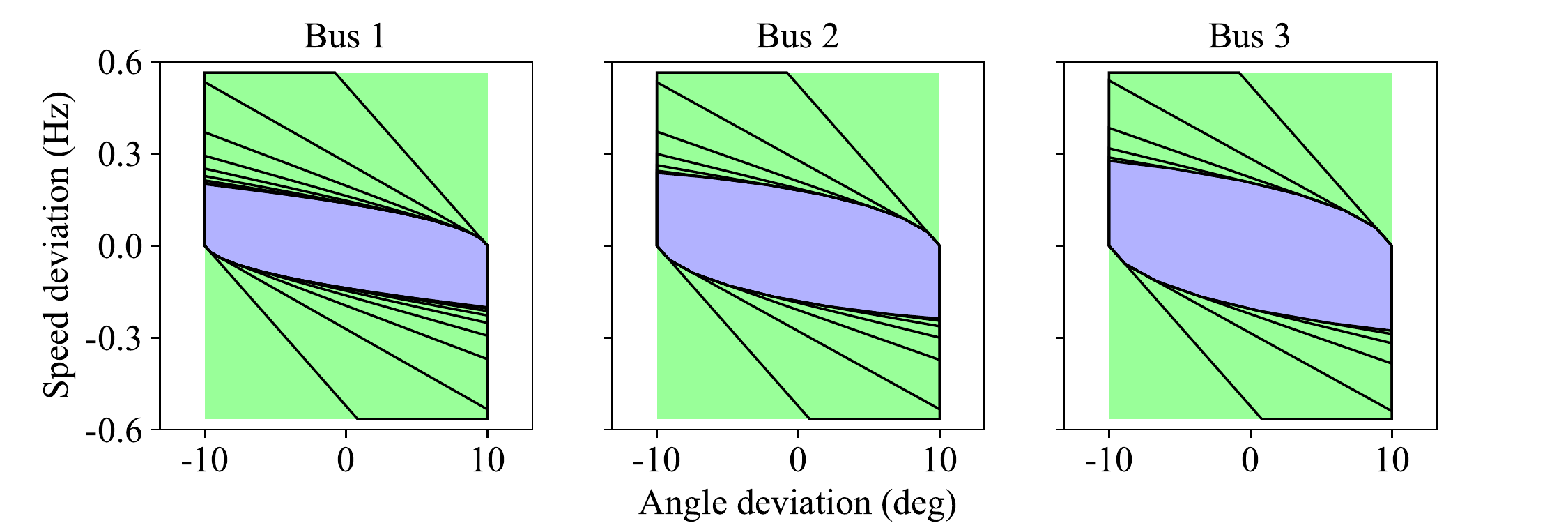}
    \caption{Safe sets (green), RCI sets (blue), and boundaries of iterated admissible sets (black). Each RCI is described in halfspace representation by 20 inequalities.}
    \label{fig:sets}
\end{figure}



Figure \ref{fig:sims} displays simulation results for bus 1. To test robustness to initial conditions, trajectories are initialized at various points distributed uniformly around the boundaries of the invariant sets. We simulate the 9 bus system for two seconds at 50-millisecond time increments, under adversarial disturbances targeting bus 1. The RMPC controller maintains robust invariance while the baseline MPC and LQR controllers are forced out of the safe sets every time. This result is further illustrated in the time domain in Figure \ref{fig:time_domain}.

\begin{figure}[ht]
    \centering
    \includegraphics[width = 9cm, height=3cm]{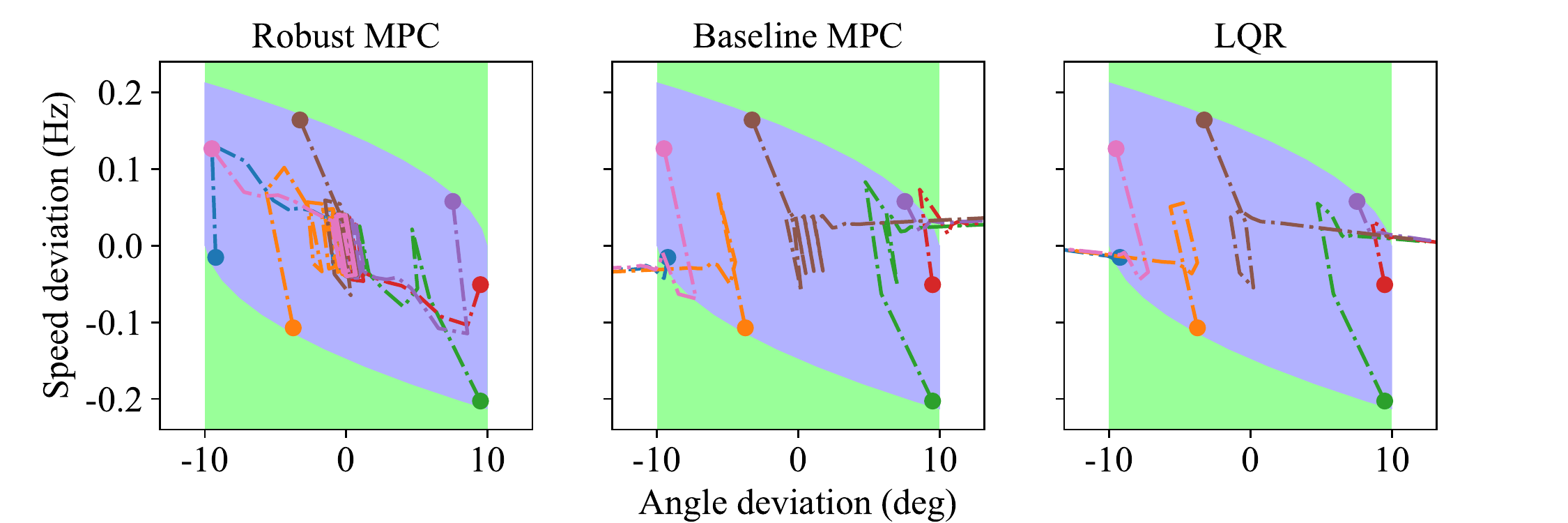}
    \caption{State trajectories at bus 1 subjected to adversarial disturbances. The different colored trajectories correspond to different initial conditions.}
    \label{fig:sims}
\end{figure}

\begin{figure}[ht]
    \centering
    \includegraphics[width = 9cm,height=3cm]{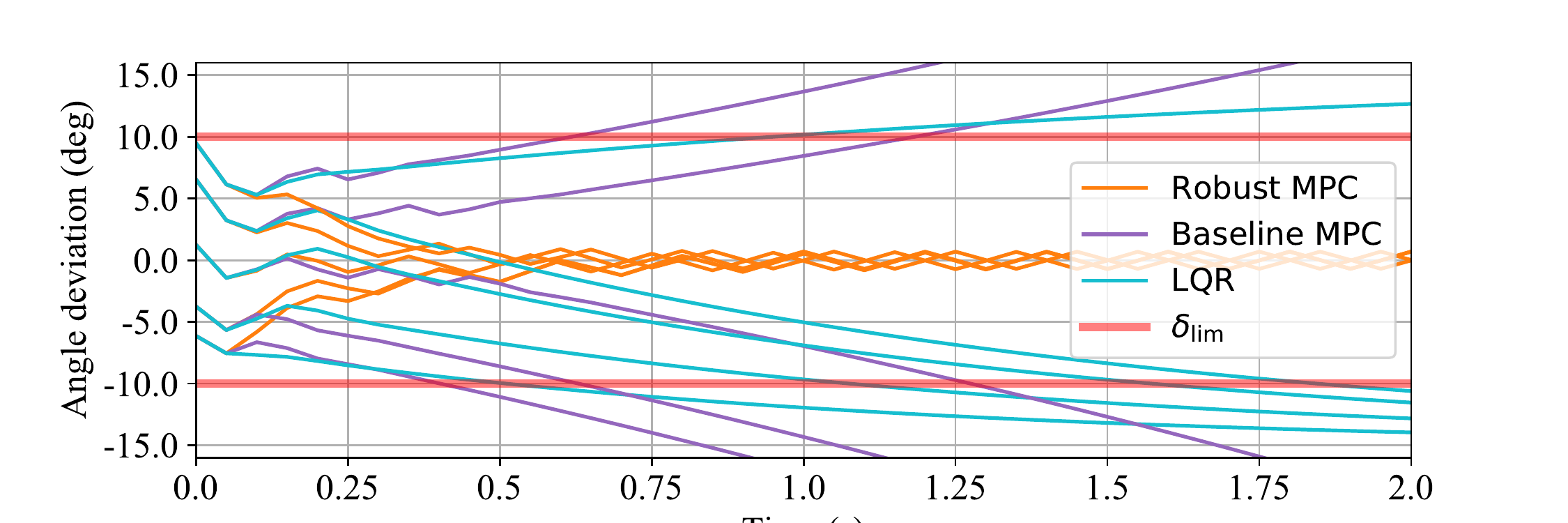}
    \caption{Time-domain trajectories of $\delta_1$ under adversarial disturbances, showing that RMPC maintains safety while the other two controllers do not.}
    \vspace{-0.5cm}
    \label{fig:time_domain}
\end{figure}

%% file: Sections/section7.tex
\section{Conclusion} 
As renewable energy proliferates, set-theoretic control is becoming increasingly relevant in power system operations. We presented a scalable, distributed procedure for generating local control laws that have global safety guarantees.  The results show that the proposed controller is robust to adversarial disturbances where conventional controllers fail easily. Future work includes tightening worst-case assumptions and modelling correlations within each disturbance class.

%% file: Sections/section8.tex
\appendix

\subsection{Power system model} \label{app:model}

\subsubsection{Elimination of algebraic buses} \label{app:model_reduction}

In order to eliminate the algebraic equations from the system dynamics, we first define the following matrices:

\begin{itemize}
    \item Let $I_G \in \R^{N_B \times N_G}$ with $[I_G]_{ij} = 1$ if the $j$th generator is located at bus $i$, and $[I_G]_{ij} = 0$ otherwise. 
    \item Let $I_L \in \R^{N_B \times N_L}$ with $[I_G]_{ij} = 1$ if the $j$th load is located at bus $i$, and $[I_L]_{ij} = 0$ otherwise. 
    \item Let $H \in \R^{N_B \times N_B}$ be a diagonal matrix with $H_{ii} = \frac{1}{C_i + K_i}$ if $i \in \mathcal{G}$ and $H_{ii} = \frac{1}{C_i}$ otherwise.
    \item Let $C \in \R^{N_B \times N_B}$ with $[C]_{ij} = C_{ij}$.
    \item Let $K \in \R^{N_G \times N_G}$ be a diagonal matrix with $[K]_{ii} = K_i$.
\end{itemize}

With the above matrix definitions, the power balance equations at the algebraic buses can be rewritten as \begin{gather}
    H^{-1} \theta = I_G u - I_L d + C \theta + I_G K \delta,
\end{gather}
where $\theta \in \R^{N_B}, u \in \R^{N_G}, d \in \R^{N_L},$ and $\delta \in \R^{N_G}$ are vectorized quantities. Solving for $\theta$ we obtain \begin{gather}
    \theta = (H^{-1}-C)^{-1}(I_G u - I_L d + I_G K \delta) \label{eq:oct1_1}
\end{gather}
Since $H$ is diagonal and $C$ is an adjacency matrix with no diagonal entries, it is reasonable to assume that $(H^{-1}-C)$ is invertible.

\subsubsection{System dynamics} \label{app:model_dynamics}

Let $R = I_{N_G \times N_G} \otimes \m{1&0}$ such that $\delta = Rx \in \R^{N_G}.$ The dynamics of each bus are described by \begin{subequations}
    \begin{equation}
    \dot{\delta}_i = \omega_i
    \end{equation}
    \begin{equation}
    \dot{\omega}_i = -\frac{K_i}{M_i} \delta_i - \frac{D_i}{M_i}\omega_i + \frac{K_i}{M_i}\theta_i.
    \end{equation}
\end{subequations} In matrix form, this yields \begin{gather}
    \dot{x}_i = A_i x_i + T_i \theta_i,
\end{gather} where \begin{gather}
    A_i = \m{0 & 1 \\ -\frac{K_i}{M_i} & -\frac{D_i}{M_i}}, \text{and} \\ 
    T_i = \m{0 \\ \frac{K_i}{M_i}}.
\end{gather} Therefore, the system dynamics can be written as \begin{gather}
    \dot{x} = Ax + T \theta, \text{ where} \label{oct1_2}\\ 
    A = \textbf{diag}(\{A_i\}_{i = 1}^{N_G}), \text{and} \nonumber \\
    T = (KM^{-1} \otimes I_{2 \times 2})(I_G^T \otimes \m{0 \\ 1}) = KM^{-1}I_G^T \otimes \m{0 \\ 1} \nonumber
\end{gather} where the last equality is due to the mixed product property of Kronecker products. The formula for $T$ comes from the definition of $T_i$ and the structure of $I_G^T$, which maps the set of all bus angles to the subset of bus angles corresponding to the generator terminals.

Now, combining \eqref{eq:oct1_1} and \eqref{oct1_2}, we can eliminate all algebraic buses and describe the system dynamics using \begin{gather}
    \dot{x} = \tilde{A}x + Bu + Ed, \text{ where} \label{oct2_1}\\ 
    \tilde{A} = A + T(H^{-1}-C)^{-1}I_GKR, \nonumber \\
    B = T(H^{-1}-C)^{-1}I_G, \text{ and} \nonumber \\
    E = T(H^{-1}-C)^{-1}I_L. \nonumber
\end{gather}

\subsection{Details on Algorithm \ref{alg:backprop}} \label{app:alg1}
Below, we provide some details and intuition behind the algorithm for generating maximal RCIs from \cite{Blanchini2015}. The procedure initializes with $\Sigma$. When $k = 0,$ $\X^{(0)}$ is equal to the safety/target set $\X \subset \R^n.$ To compute the 1-step admissible set of $\X^{(0)},$ we set $k = 0$ and perform the following operations. All operations can be carried out using e.g. \cite{Heirung2019}.

\begin{enumerate}
    \item \textbf{Set erosion:} Let $\P^{(k)} = \X^{(k)} \ominus \hat{E} \D = \{x \mid x + Ed \in \X^{(k)}, \any d \in \D \}.$ The set $\P^{(k)}$ represents the target set ``narrowed'' by the uncertainty in the exogenous disturbances.
    \item \textbf{Backpropagation:} Let $\M^{(k)} = \{(x,u): \hat{A}x + \hat{B}u \in \P^{(k)}, u \in \U\}$. If a state-control pair $(x,u)$ is in $\M^{(k)}$, then $\hat{A}x + \hat{B}u$ is in $\P^{(k)}$ which means $x^+ = \hat{A}x + \hat{B}u + \hat{E}d$ is in $X^{(k)}$ for any $d \in \D.$
    \item \textbf{Projection:} Let $\r^{(k)}$ = $\{x: \text{ there exists } u: (x,u) \in \M^{(k)}\}$. Then by Definition \ref{def:preimage}, $\r^{(k)}$ is the preimage set of $\X^{(k)}$. $\r^{(k)}$ consists of all points in $\R^n$ that can be driven into $\X^{(k)}$ in one time step, without prior knowledge of the disturbance.
    \item \textbf{Intersection:} Let $\X^{(k+1)}$ = $\r^{(k)} \cap \X$ be the $k+1$-steps admissible set, which consists of all points in $\X$ that can be driven into $\X^{(k)}$ in one time step, without prior knowledge of the disturbance.
\end{enumerate}

By induction, it is easy to see that the points in $\X^{(k+1)}$ can be made to stay in $\X$ for at least $k+1$ time steps under any disturbance sequence. Taking $k \rightarrow \infty$, it becomes obvious that $\X^{(\infty)}$ is an RCI. Moreover, $\X^{(\infty)}$ is the largest RCI contained in $\X$. One way to see this is that in the backpropagation step, the set $\M^{(k)}$ is as inclusive of state-action pairs as possible.

\subsection{Proof of Theorem \ref{thm:consensus_convergence}} \label{app:thm}

The two following lemmas are useful in the analysis of the convergence of Algorithm \ref{alg:consensus}.

\begin{lemma} \label{rem:nov7_1}
If the $\delta$ limits of the coupling disturbance sets $\Y_i^{(k,l)}$ converge, then the sets $\Y_i^{(k,l)}$ themselves converge. \end{lemma}

\begin{proof} If the $\delta$ limits of $\Y_j^{(k,l)}$ converge as $l \rightarrow \infty$ for all $j \in \N_i$, then the set $\X_i^{(k+1,l)}$ converges in $l$, since buses are only coupled through their angles. If $\X_i^{(k+1,l)}$ converges in $l$ for all $i$, then $\Y_j^{(k,l+1)}$ converges in $l$ for all $j$, by definition of $\Y_j^{(k,l+1)}$ (Alg. \ref{alg:consensus}).
\end{proof}

\begin{lemma} \label{rem:nov7_2}
If the sets $\r_i^{(k,l)}$ converge as $l \rightarrow \infty,$ then the sets $\Y_i^{(k,l)}$ converge as $l \rightarrow \infty$. Similarly, if the angle limits of $\r_i^{(k,l)}$ converge, then the angle limits of $\Y_i^{(k,l)}$ converge.
\end{lemma}
\begin{proof}
 If the sets $\r_i^{(k,l)}$ converge, then the sets $\X_i^{(k+1,l)}$ converge because intersection with $\X_i^{(k)}$ cannot increase the distance between iterates. By definition of $\Y_i^{(k,l)}$ (Alg. \ref{alg:consensus}), the sets $\Y_i^{(k,l)}$ converge. This also holds when only the angle limits of the sets are considered.
\end{proof}


Assume the dynamics are given by \eqref{eqn:nov5_1} and first consider the response of a single subsystem $\Sigma_i$ to a change in the coupling disturbances represented by $\Y_i^{(k,l)} \rightarrow \Y_i^{(k,l+1)}$. With Lemma \ref{rem:nov7_2} in mind, we will compute the difference in angle limits between $\r_i^{(k,l)}$ and $\r_i^{(k,l+1)}$. The polytope $\\X_i^{(k)}$ can be described in halfspace representation as \begin{align}
    \\X_i^{(k)} = \{x_i \mid F_i^{(k)} x_i \leq g_i^{(k)}\}, 
\end{align} and the preimage set of $\X_i^{(k)}$ under state-coupling disturbances $\Y_i^{(k,l)}$, input-coupling disturbances $\U$, and exogenous disturbances $\D$ is \begin{align}
    \r_i^{(k)} = \{&x_i \mid \exists u_i \in \U_i: F_i^{(k)} x_i^+ \leq g^{(k)}, \nonumber \\ 
    &\any y_i \in \Y_i^{(k,l)}, \any u_{\N_i} \in \U_{\N_i}, \any d \in \D \}. \label{eqn:nov5_4}
\end{align} A vertex of $\r^{(k)}$ is described by taking any two rows of the inequalities in \eqref{eqn:nov5_4} such that the corresponding rows of $F^{(k)}$ are linearly independent. Therefore, any vertex of $\r_i^{(k,l)}$ satisfies \begin{align}
    \bar{F}_i^{(k)}[\hat{A}_{1i}x_i + \hat{B}_{1i} u_i^* + \hat{A}_{2i}y_i^* + \hat{B}_{2i}u_{\N_i}^* + \hat{E}_id^*] = \bar{g}^{(k)} \label{eqn:nov5_5}
\end{align} where  $u_i^*, y_i^*, u_{\N_i}^*,$ and $d^*$ are chosen from the vertices of their corresponding sets, and $\bar{F}^{(k)}$ and $\bar{g}^{(k)}$ are formed from the appropriate rows of $F^{(k)}$ and $g^{(k)}.$ We note that $u_i$ is chosen optimally while $y_i, u_{\N_i},$ and $d$ are chosen adversarially, in terms of relaxing/tightening the inequality in \eqref{eqn:nov5_4}. Now, if $\Y_i^{(k,l+1)} = \Y_i^{(k,l)} + \Delta \Y_i$, then the corresponding vertex of $\r_i^{(k,l+1)}$ satisfies \begin{align}
    \bar{F}_i^{(k)}[&\hat{A}_{1i}(x_i+\Delta x_i) + \hat{B}_{1i} u_i^* \nonumber \\ 
    &+ \hat{A}_{2i}(y_i^* + \Delta y_i^*) + \hat{B}_{2i}u_{\N_i}^* + \hat{E}_id^*] = \bar{g}^{(k)}. \label{eqn:nov5_6}
\end{align} Subtracting \eqref{eqn:nov5_5} from \eqref{eqn:nov5_6} to obtain the difference between the vertex of $\r_i^{(k,l)}$ and the corresponding vertex of $\r_i^{(k,l+1)}$ yields \begin{align}
    \bar{F}_i^{(k)}[\hat{A}_{1i}\Delta x_i + \hat{A}_{2i}\Delta y_i^*] = 0.
\end{align} Now, multiplying by $(\bar{F}^{(k)})^{-1}$ and rearranging yields \begin{align}
    \Delta x_i = -\hat{A}_{1i}^{-1} \hat{A}_{2i} \Delta y_i^*
\end{align} The absolute change in the $\delta$-coordinate of the vertex in question is approximated by: \begin{align}
    |\Delta \delta_i| &\leq \|\Delta x_i\|_2 \\
    &= \|\hat{A}_{1i}^{-1} \hat{A}_{2i} \Delta y_i^*\|_2 \\
    &\leq \|\hat{A}_{1i}^{-1} \|_2 \cdot \|\hat{A}_{2i}\|_2 \cdot \| \Delta y_i^*\|_2 \\
    &= \|(I + hA_{1i})^{-1}\|_2 \cdot \|hA_{2i}\|_2 \cdot \| \Delta y_i^*\|_2 \\
    &:= h \alpha_i \|\Delta y_i^*\|_2 \label{eqn:nov7_1}
\end{align} where \begin{align}
    \alpha_i &= \|(I + hA_{1i})^{-1}\|_2 \cdot \|A_{2i}\|_2 \\
    &= [o(1) + o(h)] \|A_{2i}\|_2.
\end{align}
This shows that $|\Delta \delta_i|$ is influenced primarily by the size of the change in coupling disturbances, the strength of the coupling between subsystem $i$ and its neighbors, and the discretization time step. 

Now, consider the system as a whole. Suppose a change in the coupling disturbance sets $\{\Y_i^{(k,l)}\}_{i=1}^{N_G}$ is induced by incrementing $k$ or $l$, and suppose the largest such change for any vertex of $\Y_i^{(k,l)}$ for any $i$ has magnitude $\|\Delta y_i^{(l)}\|_2 = \eta \sqrt{2\nu}$, where $\eta$ is a positive constant and $\nu$ is the largest degree of any node in the system (so that the dimension of $\Delta y_i^{(l)}$ is $2 \nu$). From \eqref{eqn:nov7_1}, for each subsystem $i$, the change in angle limits between $\r_i^{(k,l)}$ and $\r_i^{(k,l+1)}$ is \begin{align}
    |\Delta \delta_i^{(l)}| &\leq h \alpha_i \eta \sqrt{2\nu} \\
    &\leq h \alpha \eta \sqrt{2\nu} 
\end{align}
where $\alpha = \max_i \alpha_i$. This change in angle limits can also be viewed as a secondary change in coupling disturbances. In the next iteration of Algorithm \ref{alg:consensus}, we have that for each $i$, \begin{align}
    \|\Delta y_i^{(l+1)}\|_2 &\leq \sqrt{2\nu} \cdot (h \alpha \eta \sqrt{2\nu}) \\
    |\Delta \delta_i^{(l+1)}| &\leq h \alpha \|\Delta y_i^{(l+1)}\|_2 \\ 
    &\leq \eta (h \alpha \sqrt{2\nu})^2
\end{align} and in general, 
\begin{align}
    |\Delta \delta_i^{(l+p)}| \leq \eta (h \alpha \sqrt{2\nu})^p,
\end{align} showing that if $h < \frac{1}{\alpha \sqrt{2\nu}}$, then $\sum_{l=0}^\infty |\Delta \delta_i^{(l)}|$ is bounded above by a convergent geometric series, and therefore converges. An approximate upper bound for $h$ is $\frac{1}{\max_i\|A_{2i}\| \sqrt{2\nu}}$. This proves that if $h$ is sufficiently small, the $\delta$ limits of the sets $\r_i^{(k,l)}$ converge as $l \rightarrow \infty$. By Lemma \ref{rem:nov7_2}, the $\delta$ limits of $\Y_i^{(k,l)}$ also converge. By Lemma \ref{rem:nov7_1}, the sets $\Y_i^{(k,l)}$ converge.
\qed